\def\year{2020}\relax
\documentclass[letterpaper]{article} % DO NOT CHANGE THIS
\usepackage{aaai20}  % DO NOT CHANGE THIS
\usepackage{times}  % DO NOT CHANGE THIS
\usepackage{helvet} % DO NOT CHANGE THIS
\usepackage{courier}  % DO NOT CHANGE THIS
\usepackage[hyphens]{url}  % DO NOT CHANGE THIS
\usepackage{graphicx} % DO NOT CHANGE THIS
\usepackage{graphicx}  % DO NOT CHANGE THIS
\usepackage{subfig}
\usepackage{amsmath}
\usepackage{amsfonts}
\usepackage{tikz}
\usepackage{amsthm}
\usepackage{lineno}
\title{HS-CAI: A Hybrid DCOP Algorithm via Combining Search with Context-based Inference}
\author{
	Dingding Chen,\textsuperscript{\rm 1}
	Yanchen Deng,\textsuperscript{\rm 2} 
	 Ziyu Chen,\textsuperscript{\rm 1,}\thanks{ Corresponding author.}
	Wenxing Zhang,\textsuperscript{\rm 1}
	Zhongshi He\textsuperscript{\rm 1}\\
	\textsuperscript{\rm 1}College of Computer Science, Chongqing University, Chongqing, China\\
	\textsuperscript{\rm 2}School of Computer Science and Engineering, Nanyang Technological University, Singapore\\
	\{dingding,chenziyu,zshe\}@cqu.edu.cn, ycdeng@ntu.edu.sg, wenxinzhang18@163.com
}

\begin{document}
\maketitle
\begin{abstract}
Search and inference are two main strategies for optimally solving Distributed Constraint Optimization Problems (DCOPs). Recently, several algorithms were proposed to combine their advantages. Unfortunately, such algorithms only use an approximated inference as a one-shot preprocessing phase to construct the initial lower bounds which lead to inefficient pruning under the limited memory budget. 
On the other hand, iterative inference algorithms (e.g., MB-DPOP) perform a context-based complete inference for all possible contexts but suffer from tremendous traffic overheads.
In this paper, $(i)$ hybridizing search with context-based inference, we propose a complete algorithm for DCOPs, named {HS-CAI} where the inference utilizes the contexts derived from the search process to establish tight lower bounds while the search uses such bounds for efficient pruning and thereby reduces contexts for the inference.
Furthermore, $(ii)$ we introduce a context evaluation mechanism to select the context patterns for the inference to further reduce the overheads incurred by iterative inferences. 
Finally, $(iii)$ we prove the correctness of our algorithm and the experimental results demonstrate its superiority over the state-of-the-art.
\end{abstract}

\section{Introduction}
Distributed Constraint Optimization Problems (DCOPs) \cite{hirayama1997distributed,fioretto2018distributed} are an elegant model for representing Multi-Agent Systems (MAS) where agents coordinate with each other to optimize a global objective. Due to their ability to capture essential MAS aspects, DCOPs can formalize various applications in the real world such as sensor network \cite{farinelli2014agent}, task scheduling \cite{maheswaran2004taking,fioretto2017multiagent}, smart grid \cite{fioretto2017distributed} and so on.

Incomplete algorithms for DCOPs \cite{zhang2005distributed,Maheswaran2006A,farinelli2008decentralised,ottens2017duct} aim to rapidly find solutions at the cost of sacrificing optimality. On the contrary, complete algorithms guarantee the optimal solution and can be generally classified into inference-based and search-based algorithms. DPOP \cite{petcu2005scalable} and Action\_GDL \cite{vinyals2009generalizing} are typical inference-based complete algorithms which employ a dynamic programming paradigm to solve DCOPs. However, they require a linear number of messages of exponential size with respect to the induced width. Accordingly, ODPOP \cite{petcu2006odpop} and MB-DPOP \cite{petcu2007mb} were proposed to trade the message number for smaller memory consumption by propagating the dimension-limited utilities with the corresponding contexts iteratively. That is, they iteratively perform a context-based inference to solve DCOPs optimally when the memory budget is limited.

Search-based complete algorithms like SBB \cite{hirayama1997distributed}, AFB \cite{gershman2009asynchronous}, PT-FB \cite{litov2017forward}, ADOPT \cite{modi2005adopt} and its variants \cite{yeoh2010bnb,gutierrez2011generalizing} perform distributed backtrack searches to exhaust the search space. They have a linear size of messages but an exponential number of messages. 
Furthermore, these algorithms only use local knowledge to update the lower bounds, which exerts a trivial effect on pruning and makes them infeasible for solving large-scale problems. Then, PT-ISABB \cite{deng2019pt}, DJAO \cite{Kim2014DJAO} and ADPOT-BDP \cite{atlas2008memory} came out to attempt to hybridize search with inference, where an approximated inference is used to construct the initial lower bounds for the search process. More specifically, PT-ISABB and ADOPT-BDP use ADPOP \cite{Petcu2005Approximations} to establish the lower bounds, while DJAO employs a function filtering technique \cite{Brito2010Improving} to get them. Here, ADPOP is an approximate version of DPOP by dropping the exceeding dimensions to ensure that each message size is below the memory limit.
However, given the limited memory budget, the lower bounds obtained in a one-shot preprocessing phase are still inefficient for pruning since such bounds cannot be tightened by considering the running contexts. That is, the existing hybrid algorithms use only a context-free approximated inference as a one-shot phase to construct the initial lower bounds. 

In this paper, we investigate the possibility of combining search with context-based inference to solve DCOPs efficiently. Specifically, our main contributions are listed as follows:
\begin{itemize}
	\item We propose a novel complete DCOP algorithm by hybridizing search with context-based inference, called HS-CAI where the search adopts a tree-based SBB to find the optimal solution and provides contexts for the inference, 
	while the inference iteratively performs utility propagation for these contexts to construct the tight lower bounds to speed up the search process.
	\item We introduce a context evaluation mechanism to extract the context patterns for the inference from the contexts derived from the search process so as to further reduce the number of context-based inferences. 
	\item We theoretically show the completeness of HS-CAI and prove that the lower bounds produced by the context-based inference are at least as tight as the ones established by the context-free approximated inference under the same memory budget. Moreover, the experimental results demonstrate HS-CAI outperforms state-of-the-art complete DCOP algorithms.
\end{itemize}

\begin{figure}	
	\centering
	\subfloat[Constraint graph]{
		\includegraphics[scale=0.21]{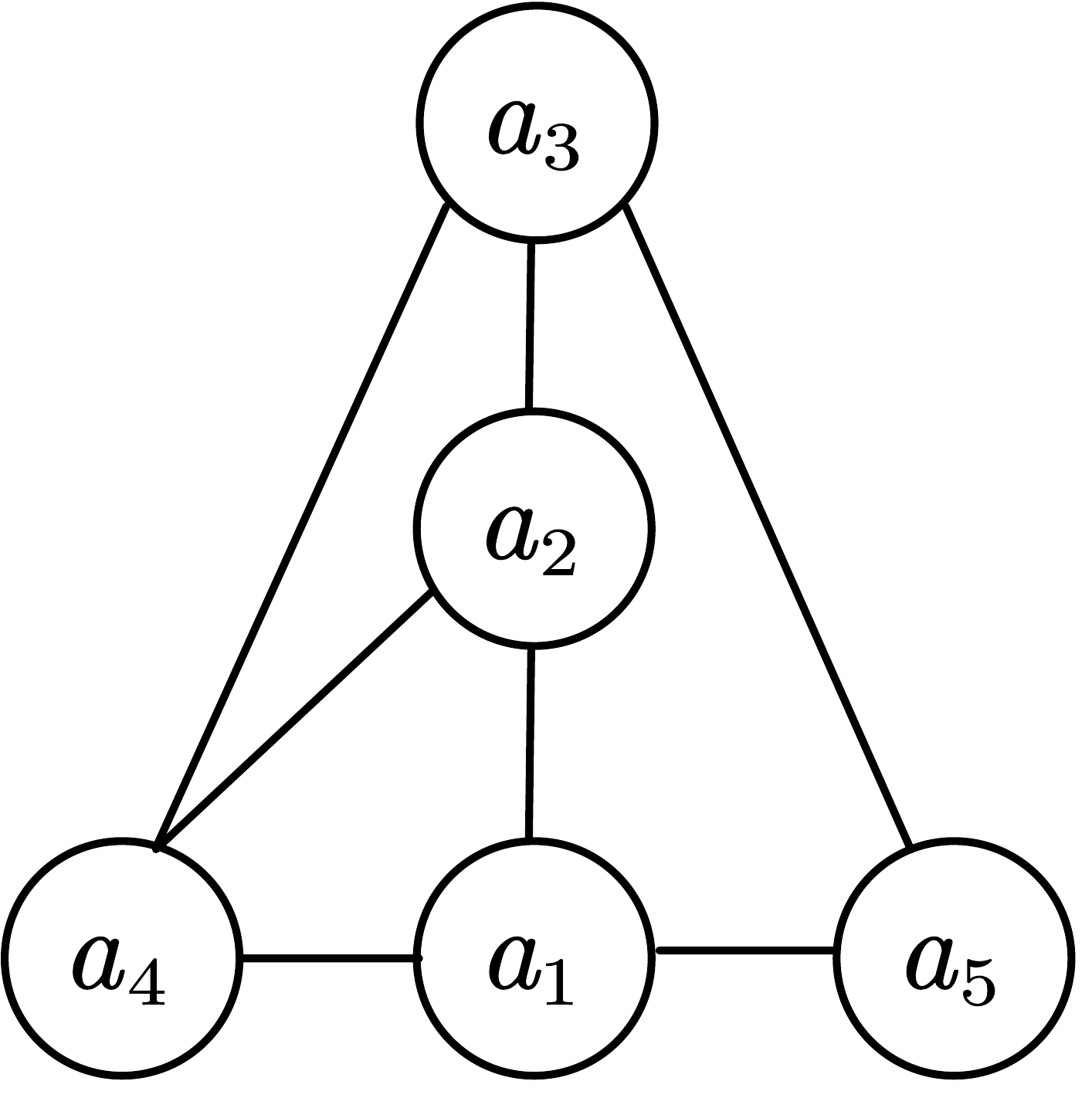} 
	}
		\subfloat[Constraint matrix]{
			\includegraphics[scale=0.75]{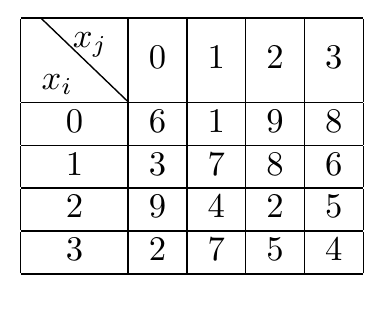} 
		}
	\subfloat[Pseudo tree]{
		\includegraphics[scale=0.14]{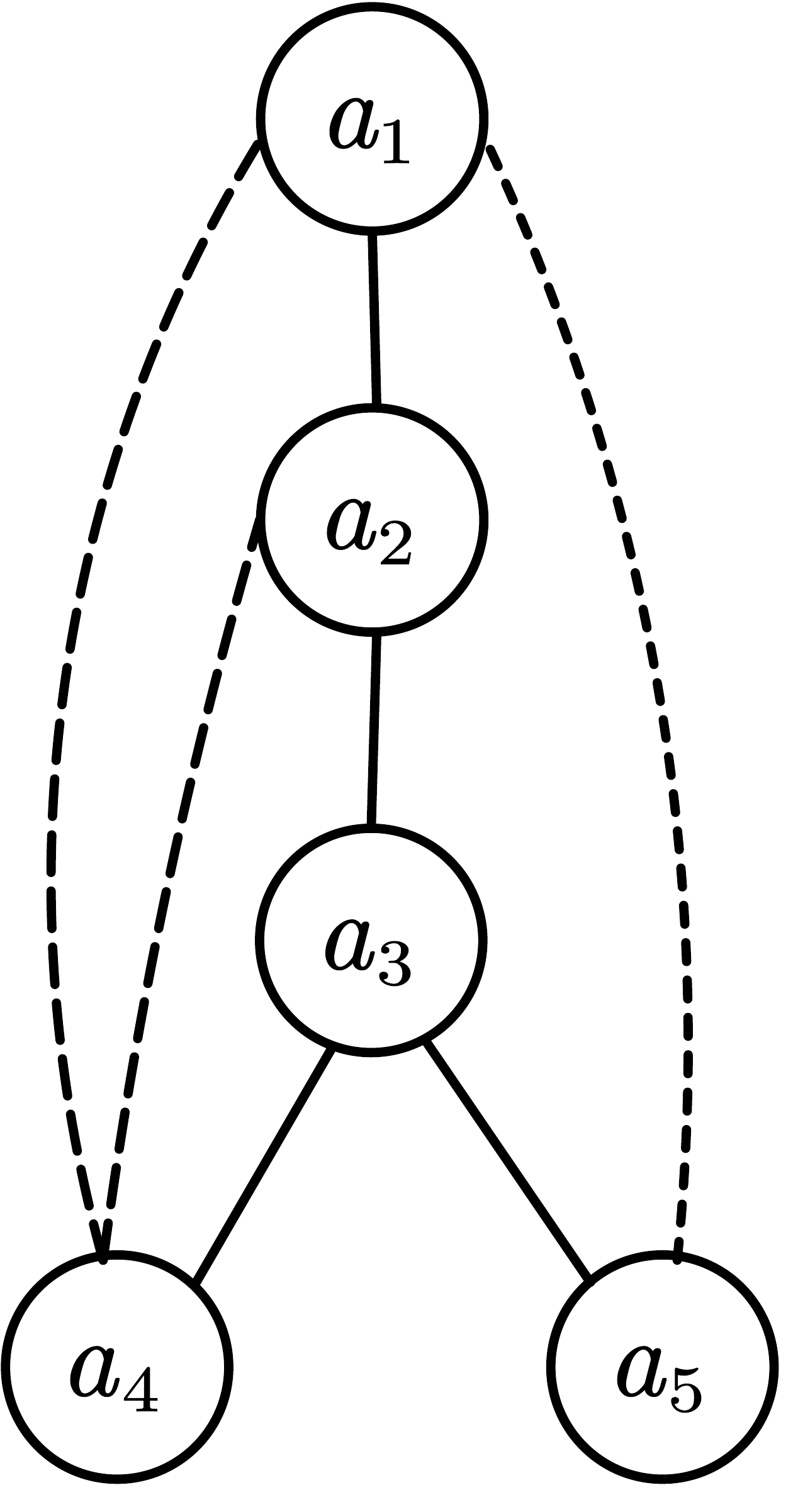} 
	}
	\centering
	\caption{An example of a DCOP and its pseudo tree}
	\label{graph}
\end{figure}

\section{Background}
In this section, we expound the preliminaries including DCOPs, pseudo tree, MB-DPOP and ODPOP.
\subsection{Distributed Constraint Optimization Problems}
A distributed constraint optimization problem \cite{modi2005adopt} can be defined by a tuple $\langle A,X,D,F\rangle$ where 
\begin{itemize}
	\item $A=\{a_1,a_2,\dots,a_n\}$ is a set of agents.
	\item $X=\{x_1,x_2,\dots,x_m\}$ is a set of variables.
	\item $D=\{D_1,D_2,\dots,D_m\}$ is a set of finite, discrete domains. Each variable $x_i$ takes a value in $D_i$.
	\item $F=\{f_1,f_2,\dots,f_q\}$ is a set of constraint functions. Each function $f_i:D_{i1}\times \cdots \times D_{ik}\rightarrow \mathbb{R}_{\geq 0}$ specifies the non-negative cost for each combination of $x_{i1},\cdots ,x_{ik}$.
\end{itemize}

For the sake of simplicity, we assume that each agent holds exactly one variable (and thus the term \emph{agent} and \emph{variable} could be used interchangeably) and all constraints are binary (i.e., $f_{ij}:D_i\times D_j\rightarrow \mathbb{R}_{\geq 0}$). A solution to a DCOP is the assignments to all the variables such that the total cost is minimized. That is,
$$\small X^*=\mathop{\arg\min}_{d_i\in D_i,d_j\in D_j}\sum_{f_{ij}\in F} f_{ij}(x_i=d_i,x_j=d_j)$$ 
A DCOP can be represented by a constraint graph where a vertex denotes a variable and an edge denotes a constraint. Fig. \ref{graph} (a) presents a DCOP with five variables and seven constraints. For simplicity, the domain size of each variable is four and all constraints are identical as shown in Fig. \ref{graph}(b). 

\subsection{Pseudo Tree}
A depth first search \cite{freuder1985taking,dechter2003constraint} arrangement of a constraint graph is a pseudo tree with the property that different branches are independent, and categorizes its constraints into tree edges and pseudo edges (i.e., non-tree edges). Thus, the neighbors of agent $a_i$ can be classified into its parent $P(a_i)$, children $C(a_i)$, pseudo parents $PP(a_i)$ and pseudo children $PC(a_i)$ based on their positions in the pseudo tree and the type edges they connect with $a_i$. For clarity, we denote all the (pseudo) parents of $a_i$ as $AP(a_i)=PP(a_i)\cup \{P(a_i)\}$ and the set of ancestors who share constraints with $a_i$ and its descendants as its separators $Sep(a_i)$ \cite{petcu2005scalable}. Fig. \ref{graph}(c) presents a possible pseudo tree deriving from Fig. \ref{graph}(a). 
\subsection{MB-DPOP and ODPOP}
MB-DPOP and ODPOP apply an iterative context-based utility propagation to aggregate the optimal global utility. 
Specifically, MB-DPOP first uses a cycle-cuts idea \cite{dechter2003constraint} on a pseudo tree to determine cycle-cut variables and groups these variables into clusters. Within the cluster, MB-DPOP preforms a bounded-memory exploration; anywhere else, the utility propagation from DPOP applies. Specifically, agents in each cluster propagate memory-bounded utilities for all the contexts of cycle-cut variables. 
As for ODPOP, each agent adopts an incremental and best-first fashion to propagate the context-based utility. Specifically, an agent repeatedly asks its children for their suggested context-based utilities until it can calculate a suggested utility for its parent during the utility propagation phase. The phase finishes after the root agent receiving enough utilities to determine the optimal global utility. 

\section{Proposed Method}
In this section, we present a novel complete DCOP algorithm which utilizes both search and inference interleaved. 
\subsection{Motivation}
It can be concluded that search can exploit bounds to prune the solution space but the pruning efficiency is closely related to the tightness of bounds. However, most of search-based complete algorithms can only use local knowledge to compute the initial lower bounds, which leads to inefficient pruning. On the other hand, inference-based complete algorithms can aggregate the global utility promptly, but their memory consumption is exponential in the induced width. Therefore, it is natural to combine both search and inference by performing memory-bounded inferences to construct efficient lower bounds for search. Unfortunately, the existing hybrid algorithms only perform an approximated inference to construct one-shot bounds in the preprocessing phase, which would lead to inefficient pruning given a limited memory budget. In fact, the bounds can be tightened by the context-based inference. That is, instead of dropping a set of dimensions, named \emph{approximated} dimensions, to stay below the memory budget, we explicitly consider the running-context assignments to a subset of approximated dimensions (i.e., the context patterns) and compute tight bounds w.r.t. the context patterns. Here, we denote an assigned subset of approximated dimensions as \emph{decimated} dimensions.

More specifically, we aim to combine the advantages of search and context-based inference to optimally solve DCOPs. Different from the existing hybrid methods, we compute tight bounds for the running contexts by performing the context-based inference for the context patterns chosen from the contexts derived from the search process.

\subsection{Proposed Algorithm}
Now, we present our proposed algorithm which consists of a preprocessing phase and a hybrid phase.

\subsubsection{Preprocessing Phase} performs a bottom-up dimension and utility propagation to accumulate the approximated dimensions and establish the initial lower bounds based on these propagated utilities for search.
Accordingly, we employ a tailored version of ADPOP with the limit $k$ to propagate the approximated dimensions and incomplete utilities (we omit the pseudo code of this phase due to the limited space). Particularly, during the propagation process, each agent $a_i$ selects the dimensions $S_i$ of its highest ancestors to approximate to make the dimension size of each outgoing utility below $k$. Then, the approximated dimensions $SList_{i}$ (i.e., the dimensions approximated by $a_i$ and its descendants) and utility $preUtil_{P(a_i)}^{i}$ sent from $a_i$ to its parent can be computed as follows: 
\begin{equation} 
\small
SList_{i}=S_i\cup ( \underset{a_c\in C( a_i )}{\cup}SList_i^{c} ) 
\label{SList}
\end{equation}
\begin{equation} 
\small
preUtil_{P(a_i)}^{i}=\underset{S_i\cup \{ x_i \}}{\min}(localUtil_i\otimes (\underset{a_c\in C(a_i)}{\otimes}preUtil_{i}^{c}))
\label{preUtil}
\end{equation}
Here, $SList_i^{c}$ and $preUtil_{i}^{c}$ are the approximated dimensions and utility received from its child $a_c\in C(a_i)$, respectively. $localUtil_i$ denotes the combination of the constraints between $a_i$ and its (pseudo) parents, i.e., 
\begin{equation} 
\small
localUtil_i=\underset{a_i\in AP(a_i)}{\otimes}f_{ij}
\end{equation}

Taking Fig. \ref{graph}(c) for example, if we set $k=1$, the dimensions dropped by $a_4$ are $\scriptsize\{x_1,x_2\}$. Thus, the approximated dimensions and utility sent from $a_4$ to $a_3$ are $\scriptsize SList_3^4=\{x_1,x_2\}$ and $\scriptsize preUtil_3^4=\underset{\{ x_1,x_2,x_4 \}}{\min}( f_{41}\otimes f_{42}\otimes f_{43} )$, respectively.

\subsubsection{Hybrid Phase} consists of the search part and context-based inference part. The search part uses a variant of SBB on a pseudo tree (i.e., a simple version of NCBB \cite{chechetka2006no}) to expand any feasible partial assignments and provides contexts for the inference part. 
By using such contexts, the inference part propagates context-based utilities iteratively to produce tight lower bounds for the search part.

Traditionally, the context-based inference is implemented by considering the assignments to the approximated dimensions (that is, the decimated dimensions are equal to the approximated dimensions). For example, MB-DPOP performs an iterative context-based inference by considering each assignment combination of cycle-cut variables. However, the approach is not a good choice for our case due to the following facts.
First, the number of the assignments of cycle-cut variables is exponential, which would incur unacceptable traffic overheads. Moreover, the propagated utilities w.r.t. a specific combination may go out of date since another inference is required as long as the assignments change, which would be very common in the search process. Therefore, it's unnecessary to perform inference for each assignment combination of the approximated dimensions.

Therefore, we consider to reduce the number of context-based inferences by making the propagated context-based utilities compatible with more contexts. Specifically, for agent $a_i$, we consider the decimated dimensions $\scriptsize PList_i\subseteq SList_i$. Then, the specific assignments to $PList_i$ serve as a context pattern and the propagated utilities will cover all the partial assignments with the same context pattern. That is, $a_i$'s context pattern can be defined by:
\begin{equation}
\small
	ctxt_{i}=\{(x_j,Cpa_{i}(x_j))| \forall x_j\in PList_i\}
\end{equation}
where $Cpa_{i}$ refers to $a_i$'s received current partial assignment, and $Cpa_{i}(x_j)$ is the current assignment of $x_j$.  

Taking agent $a_3$ in Fig. \ref{graph}(c) as an example, given the limit $k=1$, we have $SList_3=\{x_1,x_2\}$. Assume that $PList_3=\{x_1\}$. Thus, only the assignment of $x_1$ will be considered and $x_4$ (i.e., $SList_3\backslash PList_3$) will be still dropped during the context-based inference part. 
Further, assume that $ctxt_{3}=\{(x_1,0)\}$. Then, $ctxt_{3}$ covers four contexts ($\{(x_1,0),(x_2,0)\}$,$\{(x_1,0),(x_2,1)\}$,$\{(x_1,0),(x_2,2)\}$ and $\{(x_1,0),(x_2,3)\}$). Therefore, $a_3$ only need to perform inference for $ctxt_{3}$ rather than the four contexts above in this case.  

Selecting a context pattern is challenging as it offers a trade-off between the tightness of lower bounds and the number of compatible partial assignments. 
Thus, a good context pattern should comply with the following requirements. First, the good context pattern should be compatible with more contexts so as to avoid unnecessary context-based inference. In other words, these assignments are not likely to change in a short term. Second, the context pattern should result in tight lower bounds. Therefore, we propose a \emph{context evaluation mechanism} to select the context pattern according to the frequency of an assignment of each dimension in the approximated dimensions. 
In more detail, we consider the context pattern consisting of the assignments whose frequency is greater than a specified threshold $t$. Given $t$, we have $PList_i=\{x_j| Cnt_{i}(\langle x_j,Cpa_{i}(x_j)\rangle)>t, \forall x_j\in SList_i\}$ where $Cnt_{i}(\langle x_j,Cpa_{i}(x_j)\rangle)$ refers to the frequency of $Cpa_{i}(x_j)$ for $x_j$. 
With an appropriate $t$, the assignments in the context pattern could not change in a short term. 
On the other hand, if a partial assignment is hard for pruning, then there would be more assignments included in the context pattern, which guarantees the lower bound tightness. 

In addition, we introduce the variable $eval_i$ to ensure that the descendants of $a_i$ perform an inference only for a context pattern at a time. Once $a_i$ has found or received a context pattern, $eval_i$ is set to false to stop the context evaluation. And when the context pattern is incompatible with $Cpa_i$, $eval_i$ is set to true to indicate that $a_i$ can find a new context pattern.  

Next, we will detail how to implement the context evaluation mechanism and context-based inference part. Algorithm 1 presents the sketch of these procedures. We ignore the search part since it is an analogy to the tree-based branch and bound search algorithm PT-ISABB. 
\begin{figure} 
	\includegraphics[scale=0.6898]{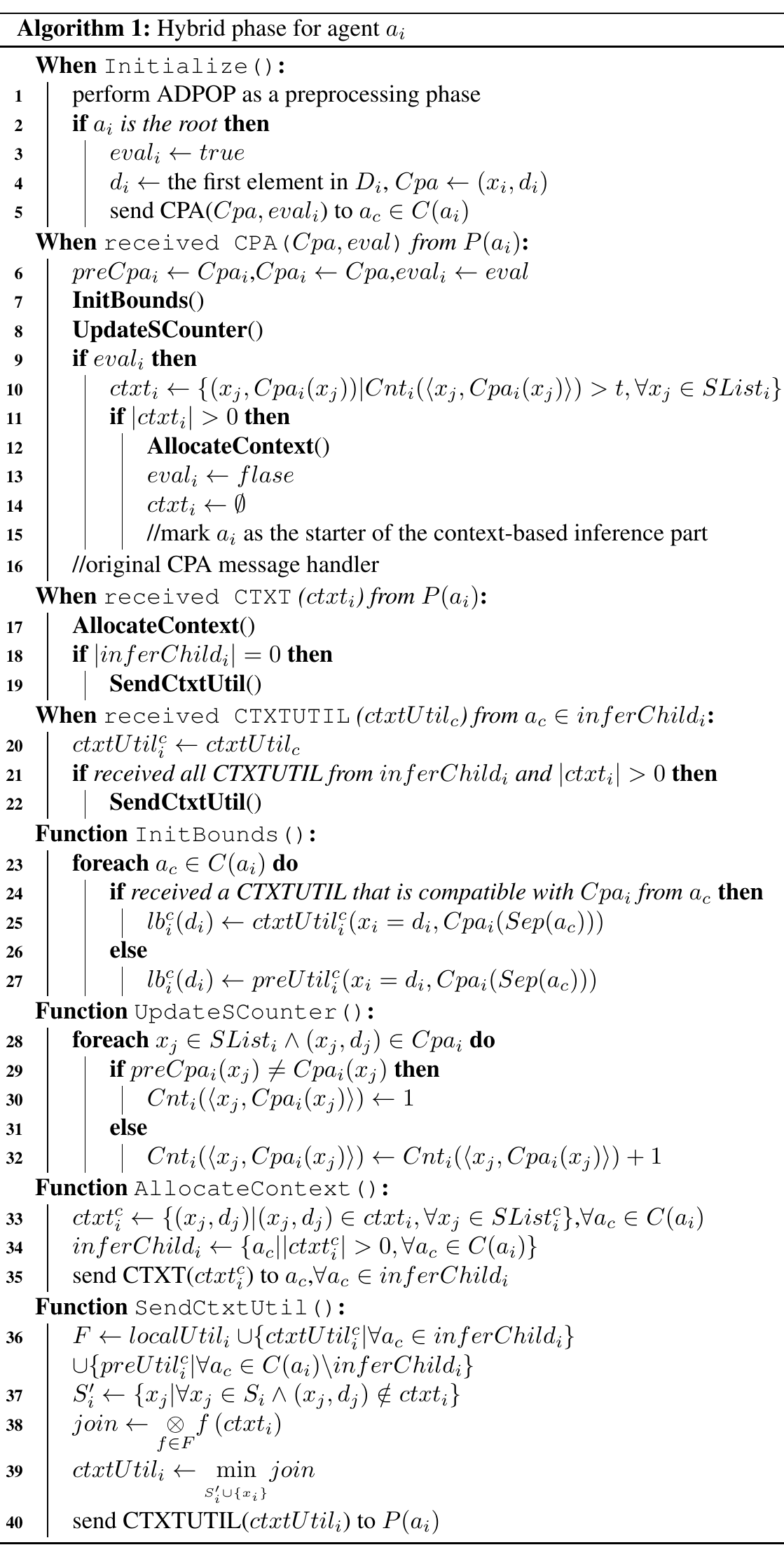} 
\end{figure}

After the preprocessing phase (line 1), the root agent starts the context evaluation by initializing $eval_i$ with true (line 2-3). Besides, it also starts the search part via CPA messages with its first assignment to its children (line 4-5).
Upon receipt of a CPA message, $a_i$ first holds the previous $Cpa_i$, and stores the received $Cpa_i$ and $eval_i$ (line 6). Afterwards, it initializes the lower bounds for each child $a_c\in C(a_i)$ according to its received utilities (line 7, 23-27). Concretely, 
the lower bound for $a_c$ is established by the context-based utility compatible with $Cpa_i$ received from $a_c$ (line 24-25). Otherwise, the bound is computed by the utility received from $a_c$ in the preprocessing phase (line 26-27). 
Next, $a_i$ updates $Cnt_i$ based on $Cpa_i$ and the previous one (line 8, 28-32). 
Specifically, for each dimension in $SList_i$, $a_i$ clears its counter if its assignment differs from its previous one (line 29-30). Otherwise, $a_i$ increases that counter (line 31-32). Then, $a_i$ finds a context pattern $ctxt_i$ if the pattern for the context-based inference has not been determined (line 9-10). 
After finding one, it allocates $ctxt_i$ and sends the allocated patterns via CTXT messages to its children who need to execute the context-based inference (i.e., $inferChild_i$) (line 11-15, 33-35). Here, $ctxt_i^c$, the allocated pattern for $a_c$, is a partition of $ctxt_i$ based on $a_c$'s approximated dimensions $SList_i^c$ (line 33).

When receiving a CTXT message, $a_i$ allocates the received pattern if there is any child who needs to perform the context-based inference (line 17, 33-35). Otherwise, it sends its context-based utility to its parent (line 18-19, 36-40). Here, its context-based utility is computed by the following steps. Firstly, it joins its local utility with the context-based utilities from $inferChild_i$ and the utilities from the other children (line 36).
Next, it applies $ctxt_i$ to fix the values of the partial dimensions in $S_i$ so as to improve the completeness of the utility (line 38). Finally, it drops the dimensions of the utility to stay below the limit $k$ (line 39). 
After all the context-based utilities from $inferChild_i$ have arrived, $a_i$ sends its context-based utility to its parent if it is not the starter of the context-based inference (line 21-22).

Considering $a_3$ in Fig. \ref{graph}(c), assume that the context pattern has not been determined and $\scriptsize Cnt_3=\{(\langle x_1,0 \rangle,1),(\langle x_2,0 \rangle,1)\}$. Given $t=1$, we have $\scriptsize Cnt_3=\{(\langle x_1,0 \rangle,2),(\langle x_2,1 \rangle,1)\}$ and $ctxt_3=\{(x_1,0)\}$ after $a_3$ receives a CPA with $\scriptsize \{(x_1,0),(x_2,1)\}$.
 Since the approximated dimensions for its child $a_4$ are $\scriptsize \{x_1,x_2\}$, $a_3$ sends a CTXT message with the context pattern $\scriptsize \{(x_1,0)\}$ to $a_4$. 
When receiving the pattern $\scriptsize \{(x_1,0)\}$, $a_4$ sends the context-based utility $\scriptsize ctxtUtil_3^4=\underset{\{x_2,x_4 \}}{\min}( f_{41}(x_1=0)\otimes f_{42}\otimes f_{43} )$ to $a_3$. Then, $a_3$ uses $ctxtUtil_3^4$ to compute the lower bound for $a_4$ after receiving the CPA message with $\scriptsize \{(x_1,0),(x_2,2)\}$ or $\scriptsize \{(x_1,0),(x_2,3)\}$. 
%And when receiving the CPA message with the assignments that contains $\{(x_1,1)\}$, the context pattern $\{x_1=0\}$ is out of the date and then $a_3$ and its children can find a new context pattern for the context-based inference part. 

\section{Theoretical Results}
In this section, we first prove the effectiveness of the context-based inference on HS-CAI, and further establish the completeness of HS-CAI. 
Finally, we give the complexity analysis of the proposed method.
	\newtheorem{lemma}{Lemma}

	\newtheorem{theorem}{Theorem}

\subsection{Lower Bound Tightness}
\begin{lemma}
	\label{lemmaBound}
	For a given $Cpa_i$, the lower bound $lb_{i}^c(d_i)$ of $a_c\in C(a_i)$ for $d_i$ produced by the context-based utility ($ctxtUtil_i^c$) is at least as tight as the one established by the utility ($preUtil_i^c$).
	That is, $ctxtUtil_i^c(PA)\geq preUtil_i^c(PA)$, where $PA=Cpa_{i}(Sep(a_c))\cup \{(x_i,d_i)\} $.
\end{lemma}
\begin{proof}
	Directly from the pseudo code, $S_c^\prime$, the dimensions dropped by $a_c$ in the context-based inference part (line 37), can be defined by:\\
	$$S_c^\prime = \{x_j|(x_j,d_j)\notin ctxt_c, \forall x_j\in S_c\}$$
	where $ctxt_c$ is the context pattern for $a_c$'s context-based inference, and $S_c$ is the dimensions dropped by $a_c$ in the preprocessing phase. Since $a_i$ has received $ctxtUtil_i^c$ from $a_c$, we have $|ctxt_c|>0$ (line 34-35, 20-21). Thus, $S_c^\prime\subset S_c$ is established.
	
	Next, we will prove Lemma \ref{lemmaBound} by induction. 
	
	Base case. $a_i$'s children are leaf agents. For each child $a_c\in C(a_i)$, we have\\
	$$
		\scriptsize
	\begin{aligned}
	&ctxtUtil_i^c(PA)
	=(\underset{S_{c}^{\prime}\cup \{ x_c \}}{\min}localUtil_c) (PA) \\
	&=\underset{x_c}{\min}(\sum_{x_j\in AP(a_c ) \backslash S_{c}^{\prime}}{f_{cj}(x_c,d_j )}+\sum_{x_j\in S_{c}^{\prime}}{\underset{x_j}{\min}f_{cj}(x_c,x_j )})
	\\
	&=\underset{x_c}{\min}(\sum_{x_j\in AP(a_c) \backslash S_c}{f_{cj}(x_c,d_j)}
	+\sum_{x_j\in S_c}{\underset{x_j}{\min}f_{cj}(x_c,x_j)}\\
	&+\sum_{x_j\in S_c\backslash S_{c}^{\prime}}{(f_{cj}(x_c,d_j) -\underset{x_j}{\min}f_{cj}(x_c,x_j))})\\
	&\geq\underset{x_c}{\min}(\sum_{x_j\in AP(a_c) \backslash S_c}{f_{cj}(x_c,d_j)}+\sum_{x_j\in S_c}{\underset{x_j}{\min}f_{cj}(x_c,x_j)}) 
	\\
	&=(\underset{S_{c}\cup \{ x_c \}}{\min}localUtil_c) (PA)=preUtil_i^c(PA)\\
	\end{aligned}
	$$
	where $d_j$ is the assignment of $x_j$ in $PA$. The equation in the third to the fourth step holds since $S_c^\prime\subset S_c$. Thus, we have proved the basis.
	
	Inductive hypothesis. Assume that the lemma holds for all $a_i$'s children. Next, we are going to show the lemma holds for $a_i$ as well. For each $a_c\in C(a_i)$, we have\\		
	$$
	\scriptsize
	\begin{aligned}
	\scriptsize
	&ctxtUtil_i^c(PA)
	=(\underset{S_{c}^{\prime}\cup \{ x_c \}}{\min}(localUtil_c+\sum_{a_{c^{\prime}}\in inferChild_c} {ctxtUtil_{c}^{c^{\prime}}}\\
	&+\sum_{a_{c^{\prime}}\in C(a_c )\backslash inferChild_c}{preUtil_{c}^{c^{\prime}}})) (PA) \\
	&\geq(\underset{S_{c}\cup \{ x_c \}}{\min}(localUtil_c+\sum_{a_{c^{\prime}}\in inferChild_c} {ctxtUtil_{c}^{c^{\prime}}}\\
	&+\sum_{a_{c^{\prime}}\in C(a_c )\backslash inferChild_c}{preUtil_{c}^{c^{\prime}}})) (PA) \\
	&\geq (\underset{S_{c}\cup \{ x_c \}}{\min}(localUtil_c+\sum_{a_{c^{\prime}}\in C(a_c )}{preUtil_{c}^{c^{\prime}}}
	) )(PA) \\
	&=preUtil_i^c(PA)\\
	\end{aligned}
	$$
	where $inferChild_c$ are $a_c$'s children who need to perform the context-based inference. Thus, Lemma \ref{lemmaBound} is proved.
\end{proof}	
\subsection{Correctness}
\begin{lemma}
	\label{lemmaCorrect}
	Given the optimal solution $X^*$, $cost_c(X^*)$, the cost to the sub-tree rooted at $a_c\in C(a_i)$ is no less than the lower bound $lb_i^c(X^*(x_i))$. That is, $cost_c(X^*)\geq lb_i^c(X^*(x_i))$.
\end{lemma}
\begin{proof}
	Since we have proved the lower bounds constructed by the context-based inference part are at least as tight as the ones established by the preprocessing phase in Lemma \ref{lemmaBound}, to prove the lemma, it is sufficient to show that $cost_c(X^*)\geq ctxtUtil_i^c(X^*)$.
	
	Next, we will prove Lemma \ref{lemmaCorrect} by induction as well. 
	
	Base case. $a_i$'s children are leaf agents. For each child $a_c\in C(a_i)$, we have
	$$
	\scriptsize
	\begin{aligned}
	&cost_c(X^*)=\sum_{a_j\in AP(a_c )}{f_{cj}(d_c^*,d_j^* )}\\
	&=localUtil_c( X^{*} )\\
	&\geq (\underset{S_c^{\prime}\cup \{x_c\} }{\min}localUtil_c)(X^*)\\
	&=ctxtUtil_i^c(X^*)\\		
	\end{aligned}
	$$
	where $d_l^*$ is the assignment of $x_l$ in $X^*$, and $S_c^\prime$ is the dimensions dropped by $a_c$ in the context-based inference part. Thus, the basis is proved.

	Inductive hypothesis. Assume the lemma holds for all $a_c\in C(a_i)$. Next, we will prove the lemma also holds for $a_i$. For each child $a_c\in C(a_i)$, we have\\
	$$
	\scriptsize
	\begin{aligned}
	&cost(Spa_c^*) =\sum_{a_j\in AP(a_c)}{f_{cj}(d_c^*,d_j^*)}+\sum_{a_{c^\prime}\in C(a_c)}{cost_{c^\prime}(X^*)}	\\
	&=localUtil_c( X^{*} )+\sum_{a_{c^\prime}\in C(a_c)}{cost_{c^\prime}(X^*)}\\
	&\geq localUtil_c( X^{*} ) +\sum_{a_{c^\prime}\in C( a_c )}{ctxtUtil_c^{c^\prime}( X^{*} )}\\	
	&=(localUtil_c +\sum_{a_{c^\prime}\in C( a_c )}{ctxtUtil_c^{c^\prime})( X^{*} )}\\
	&\geq(\underset{S_{c}^{'}\cup \left\{ x_c \right\}}{\min} (localUtil_c +\sum_{a_{c^\prime}\in C( a_c )}{ctxtUtil_c^{c^\prime}))( X^{*} )}\\
	&\geq(\underset{S_{c}^{'}\cup \left\{ x_c \right\}}{\min} (localUtil_c +\sum_{a_{c^\prime}\in inferChild_c}{ctxtUtil_c^{c^\prime}}\\
	&+\sum_{a_{c^\prime}\in C( a_c )\backslash inferChild_c}{preUtil_c^{c^\prime}))} ( X^{*})\\
	&=ctxtUtil_i^c(X^*)\\
	\end{aligned}
	$$
	Thus, the lemma is proved.
\end{proof}
\begin{theorem}
	HS-CAI is complete.
\end{theorem}
\begin{proof}
	Immediately from Lemma \ref{lemmaCorrect}, the optimal solution will not be pruned in HS-CAI. Furthermore, it has been proved that each agent will not receive two identical $Cpa$s in the search part from PT-ISABB \cite{deng2019pt}, and the termination of HS-CAI relies on the search part. Thus, HS-CAI is complete.
\end{proof}

\subsection{Complexity}
When it performs a context-based inference, $a_i$ needs to store the context-based utilities and utilities received from all its children. Thus, the overall space complexity is $O(|C(a_i)|d_{max}^{k})$ where $d_{max}=\max_{a_j\in Sep(a_i)}|D_j|$, and $k$ is the maximum dimension limit. Since a CTXTUTIL message only contains a context-based utility, its size is $O(d_{max}^{k})$.
For a CPA message, it is composed of the assignment of each agent and a context evaluation flag. Thus, the size of a CPA message is $O(|A|)$. Other messages like CTXT only carry an assignment combination of the approximated dimensions, it only requires $O(|A|)$ space.

The preprocessing phase in HS-CAI only requires $|A|-1$ messages, since only the utility propagation are performed. 
For the search part and context-based inference part in the hybrid phase, the message number of the search part grows exponentially to the agent number with the same as the search-based complete algorithms.
And the message number of the context-based inference part is proportional to the number of the context patterns selected by the context evaluation mechanism. 
\begin{figure}
	\subfloat[$k=6$]{
		\includegraphics[scale=0.215]{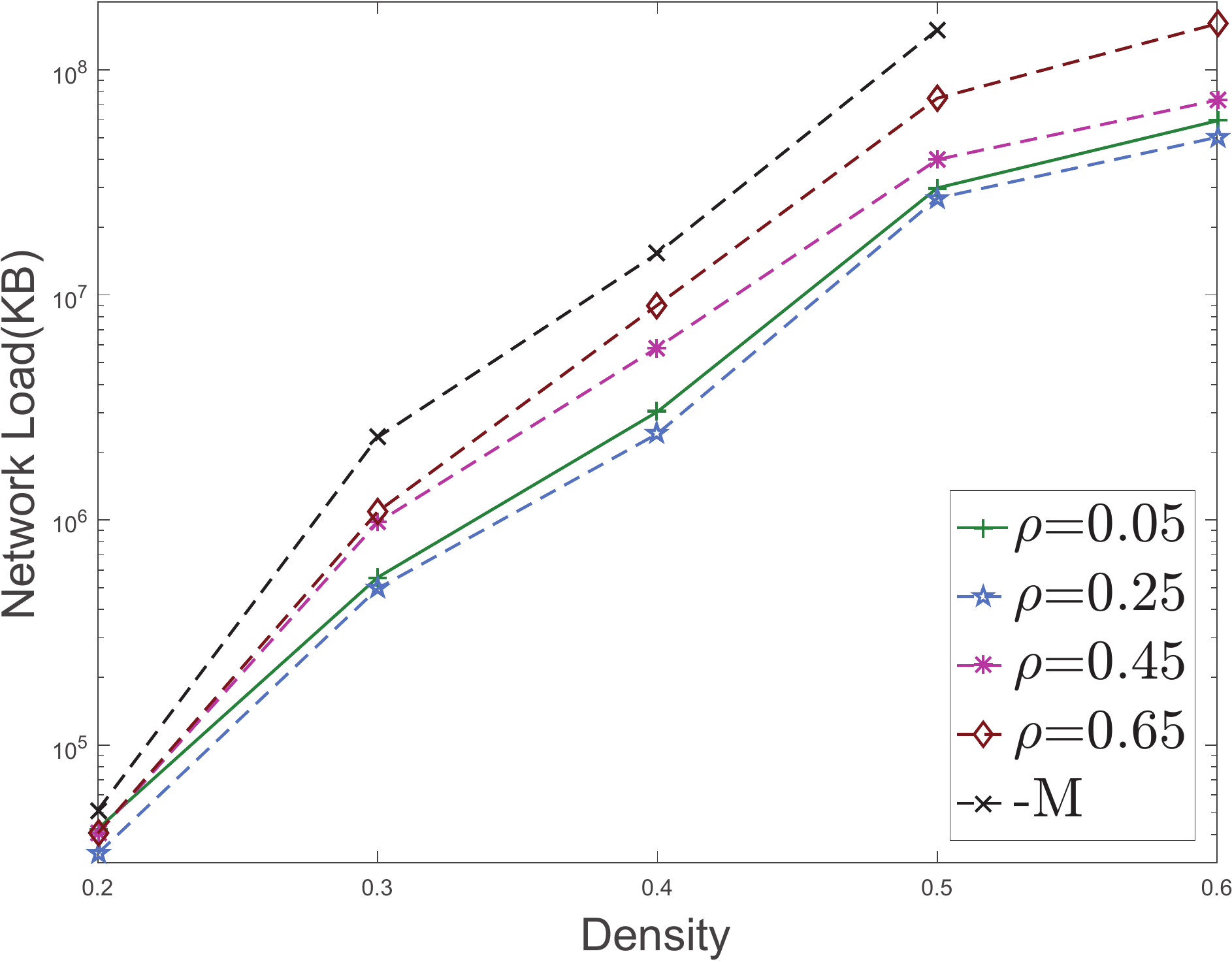} 
	}
	\subfloat[$k=10$]{
		\includegraphics[scale=0.215]{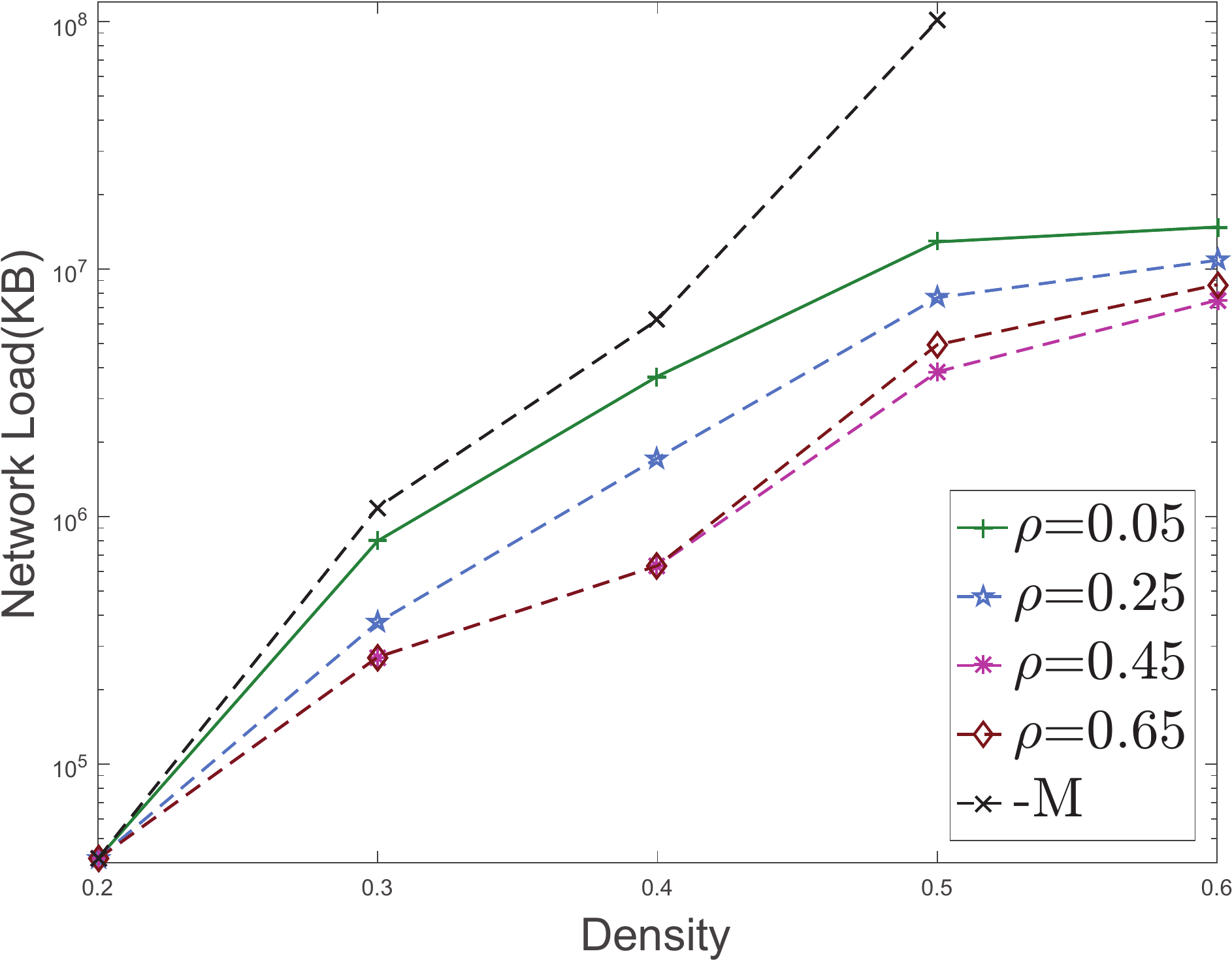} 
	}
	\caption{Network load of varying $\rho$ on different densities}
	\label{param}
\end{figure}
\begin{figure*}	
	\centering
	\subfloat[Message Number]{
		\includegraphics[scale=0.245]{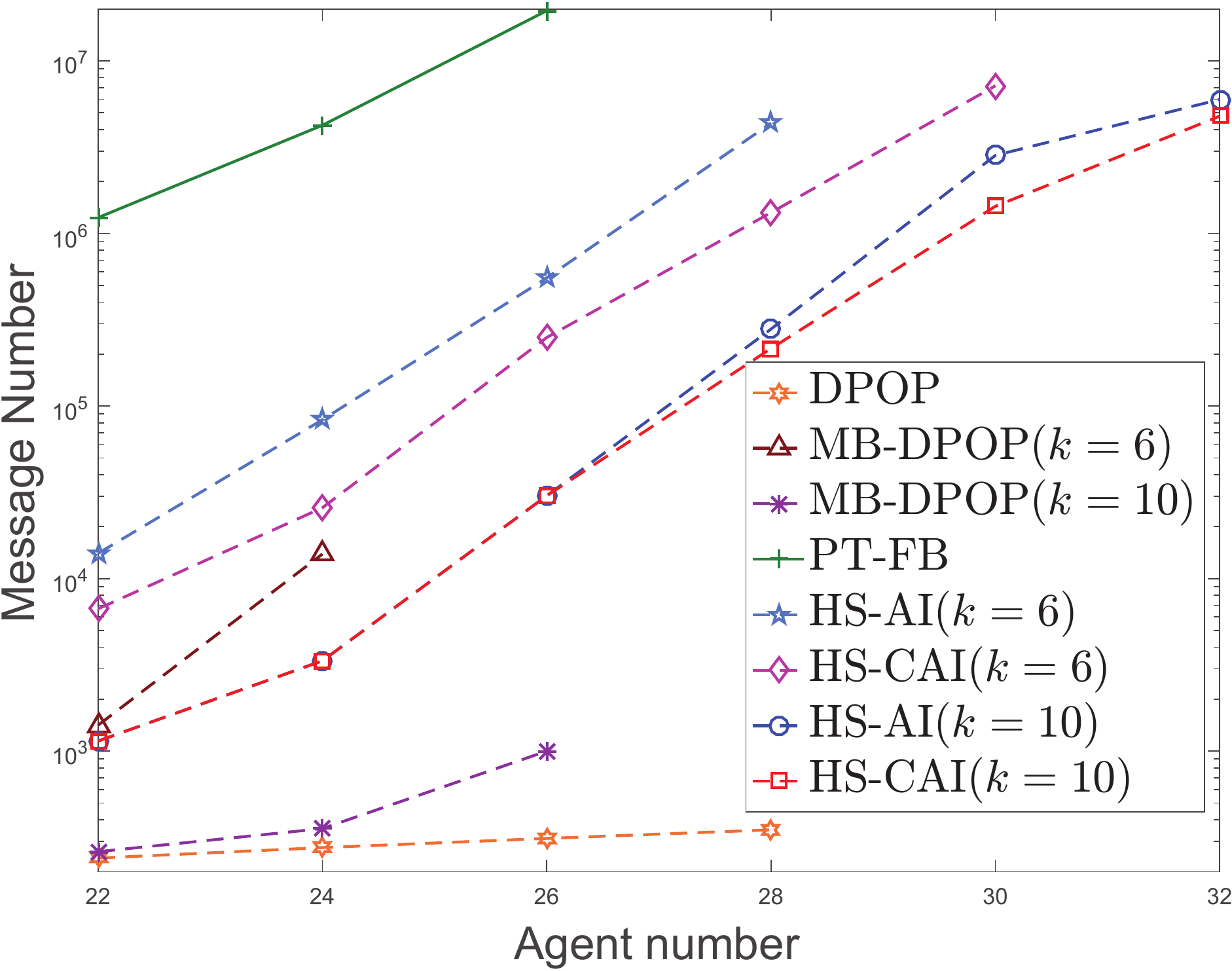} 
	}
	\subfloat[Network Load]{
		\includegraphics[scale=0.245]{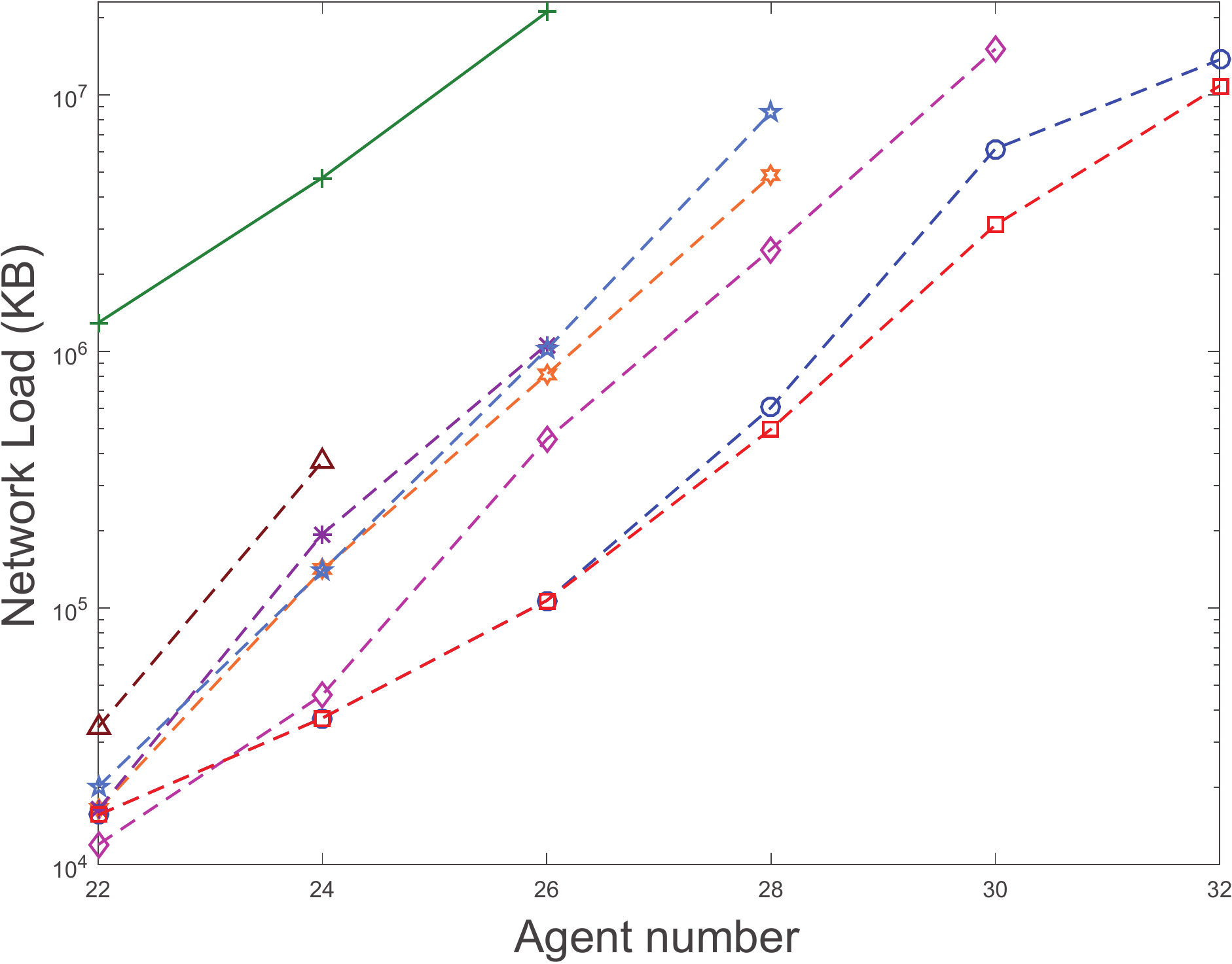} 
	}
	\subfloat[NCLOs]{
		\includegraphics[scale=0.2485]{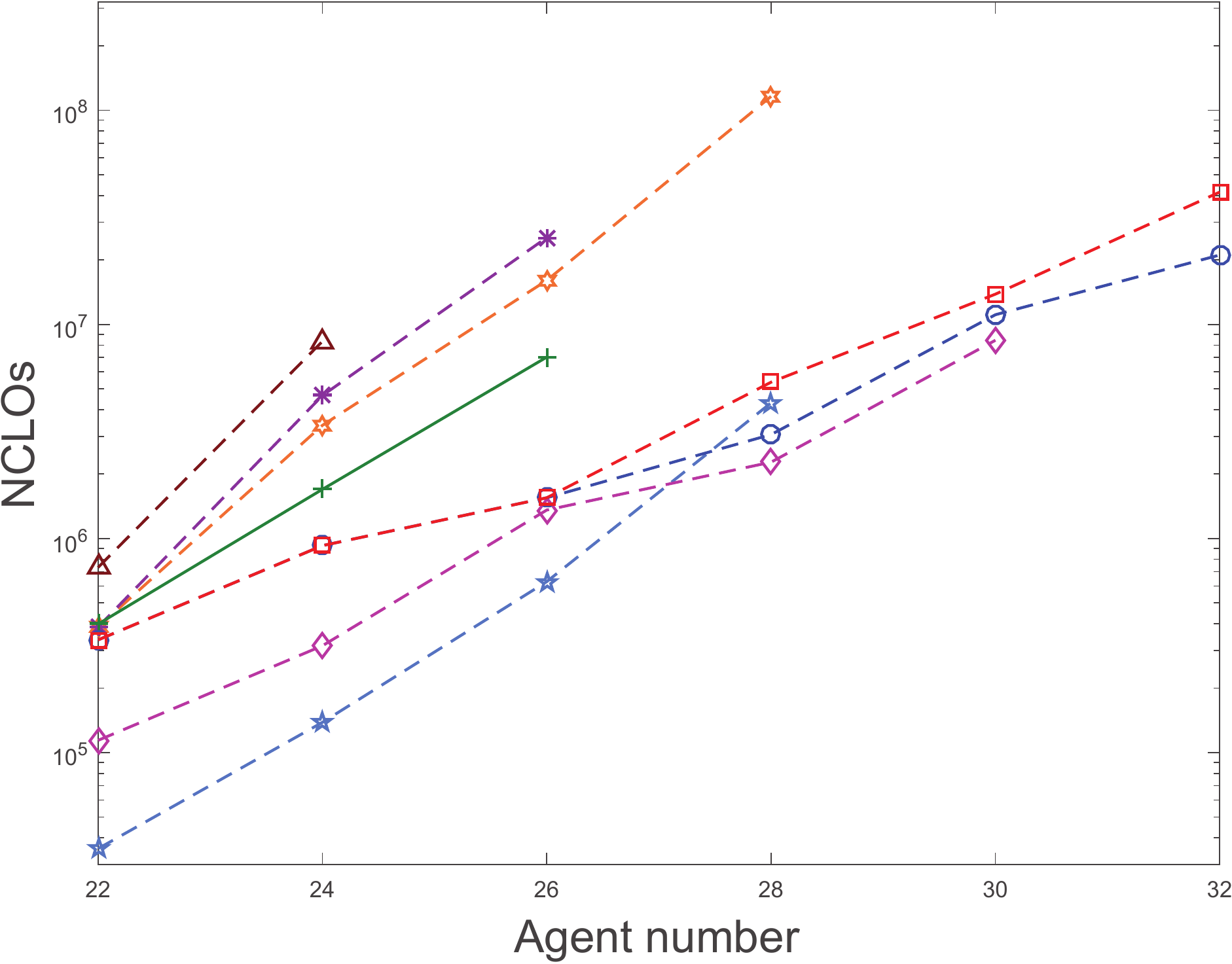} 
	}
	\\
	\caption{Performance comparison under different agents on sparse configuration}
	\label{025agent}
	\subfloat[Message Number]{
		\includegraphics[scale=0.245]{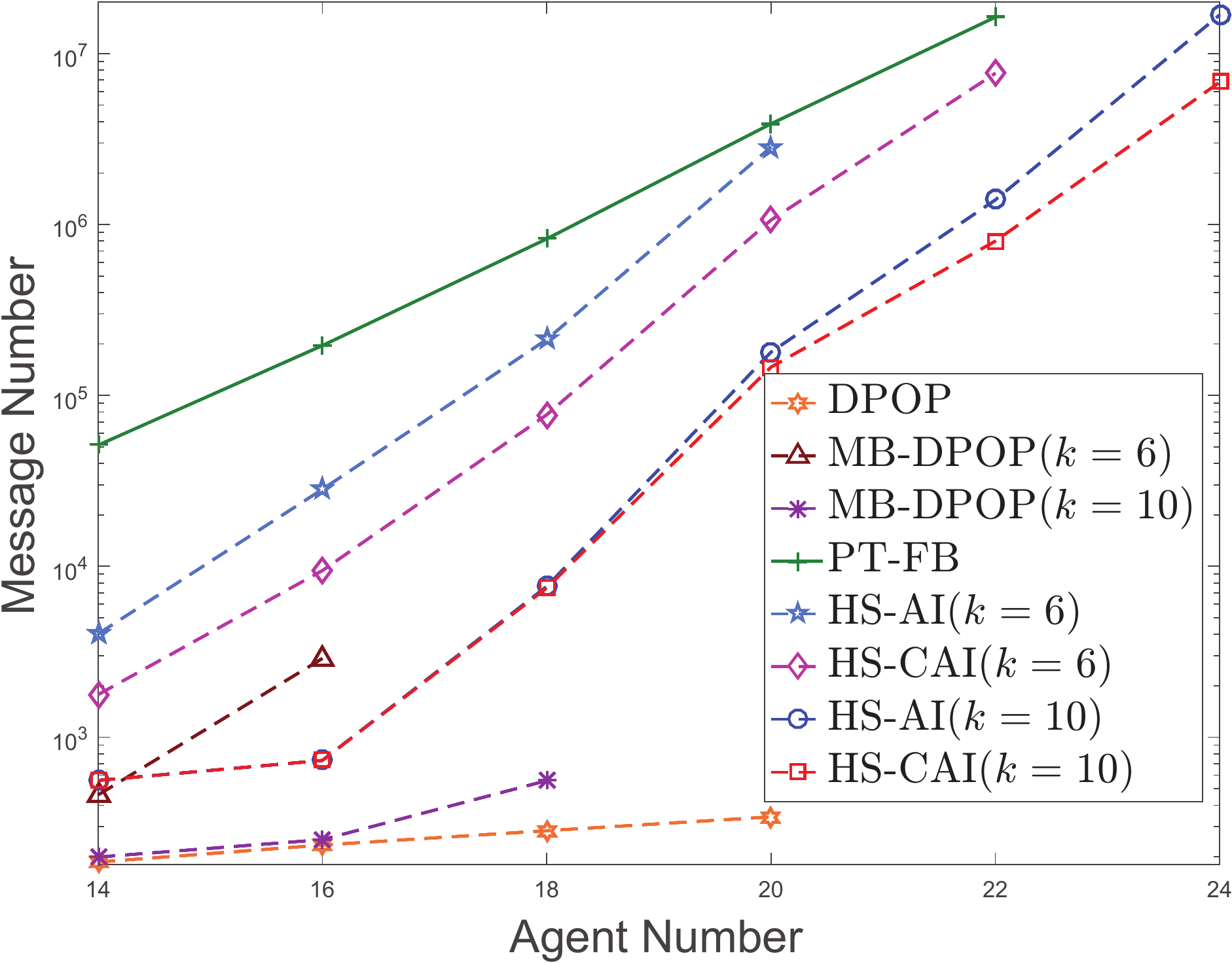} 
	}
	\subfloat[Network Load]{
		\includegraphics[scale=0.245]{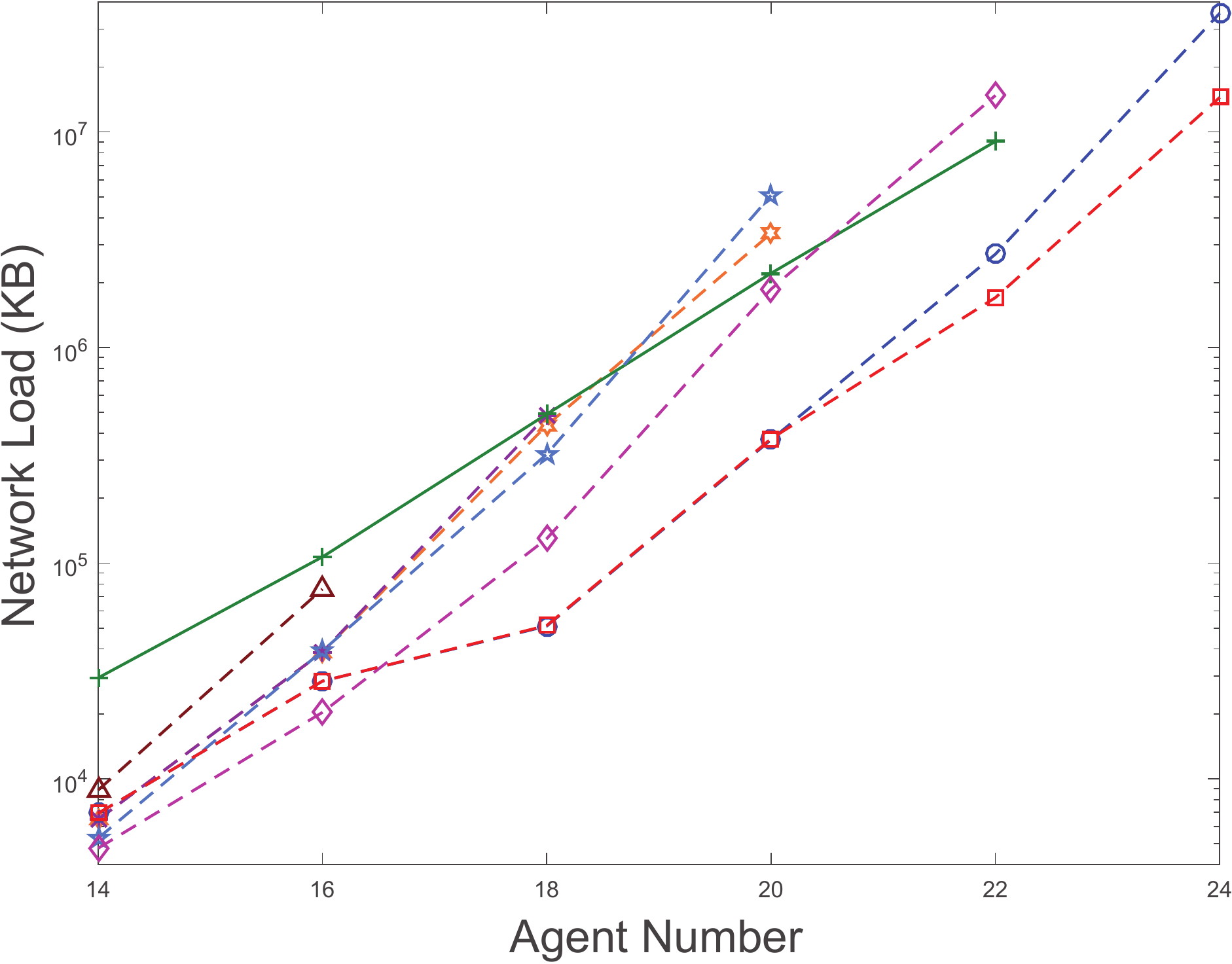} 
	}
	\subfloat[NCLOs]{
		\includegraphics[scale=0.245]{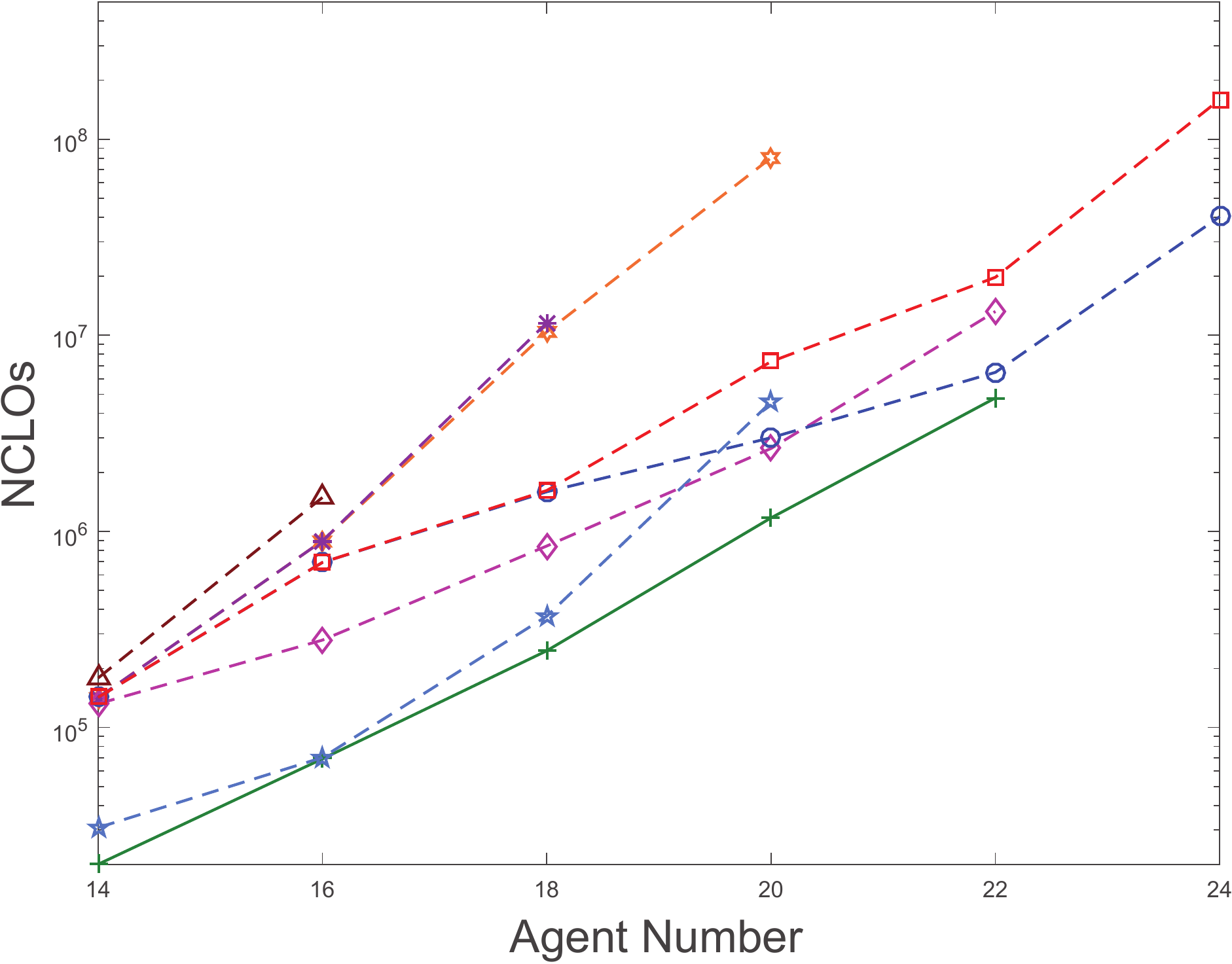} 
	}
	
	\caption{Performance comparison under different agents on dense configuration}
	\label{06agent}
\end{figure*}
\section{Empirical Evaluation}
In this section, we first investigate the effect of the parameter $t$ in the context evaluation mechanism on HS-CAI. 
Then, we present the experimental comparisons of HS-CAI with state-of-the-art complete DCOP algorithms. 
\subsection{Configuration and Metrics}
We empirically evaluate the performance of HS-CAI and state-of-the-art complete DCOP algorithms including PT-FB, DPOP and MB-DPOP on random DCOPs. Besides, we consider HS-CAI without the context-based inference part as HS-AI and HS-CAI without the context evaluation mechanism as HS-CAI(-M). Here, HS-AI is actually a variant of PT-ISABB in DCOP settings. All evaluated algorithms are implemented in DCOPSovler\footnote{ https://github.com/czy920/DCOPSovler}, the DCOP simulator developed by ourselves. Besides, we consider the parameter $t$ in HS-CAI related to both $h$ and $d_{max}$, where $d_{\max}={\max}_{a_i\in A}|D_i|$ and $h$ is the height of a pseudo tree. Therefore, we set $t=(d_{max}) ^{\rho h}$. Moreover, we choose $k=6$ and $k=10$ as the low and high memory budget for MB-DPOP, HS-AI, HS-CAI(-M) and HS-CAI. 
In our experiments, we use the message number and network load (i.e., the size of total information exchanged) to measure the communication overheads, and the NCLOs \cite{netzer2012concurrent} to measure the hardware-independent runtime where the logical operations in the inference and the search are accesses to utilities and constraint checks, respectively. For each experiment, we generate 50 random instances and report the average of over all instances.

\subsection{Parameter Tuning}
Firstly, we aim to examine the effect of different $\rho$ on the performance of HS-CAI to verify the effectiveness of the context evaluation mechanism. Specifically, we consider the DCOPs with 22 agents and the domain size of 3. The graph density varies from 0.2 to 0.6 and $\rho$ varies from 0.05 to 0.65. Here, we do not show the experiment results of $\rho$ greater than 0.65 since the larger $\rho$ leads to the exact same results as $\rho$ with 0.65.
Fig. \ref{param} presents the network load of HS-CAI(-M) and HS-CAI with different $\rho$. The average induced widths in this experiment are 8 $\sim$ 16. 
It can be seen from the figure that HS-CAI requires much less network load than HS-CAI(-M). 
That is because HS-CAI performs inference only for the context patterns selected by the context evaluation mechanism rather than all the contexts as HS-CAI(-M) does.

Besides, given the memory budget limit $k$, it can be observed that HS-CAI does not decrease all the time with the increase of $\rho$. This is due to the fact that increasing $\rho$ which leads to large $t$ can decrease the number of context-based inferences but also loose the tightness of lower bounds to some degree. Exactly as mentioned above, the context pattern selection offers a trade-off between the tightness of lower bounds and the number of compatible partial assignments. Moreover, it can be seen that the best value of $\rho$ is close to 0.25 in HS-CAI($k=6$) while the one in HS-CAI($k=10$) is near to 0.45. Thus, we choose $\rho$ to 0.25 for HS-CAI($k=6$) and 0.45 for HS-CAI($k=10$) in the following comparison experiments.
\subsection{Performance Comparisons}
Fig. \ref{025agent} gives the experimental results under different agent numbers on the sparse configuration where we consider the graph density to 0.25, the domain size to 3 and vary the agent number from 22 to 32. 
Here, the average induced widths are 9 $\sim$ 17. 
It can be seen from Fig. \ref{025agent}(a) and (b) that although the hybrid complete algorithms (e.g., HS-AI and HS-CAI) and PT-FB all use the search strategy to find the optimal solution, HS-AI and HS-CAI are superior to PT-FB in terms of the network load and message number. This is because the lower bounds in PT-FB cannot result in effective pruning by only considering the constraints related to the assigned agents.
Also, given a fixed $k$, HS-CAI requires fewer messages than HS-AI since the lower bounds produced by the context-based inference are tighter than the ones established by the context-free approximated inference. 
Besides, it can be seen that HS-CAI($k=6$) can solve larger problems than HS-AI($k=6$) and the inference-based complete algorithms like DPOP and MB-DPOP, which demonstrates the superiority of hybridizing search with context-based inference when the memory budget is relatively low.

Although inference requires larger messages than search, it can be observed from Fig. \ref{025agent}(b) that HS-CAI incurs less network load than HS-AI, which indicates that HS-CAI can find the optimal solution with fewer messages owing to the effective pruning and context evaluation mechanism.
Moreover, we can see from Fig. \ref{025agent}(c) that when solving problems with the agent number to 28, HS-CAI($k=6$) requires fewer NCLOs than HS-AI($k=6$) in spite of the exponential computation overheads incurred by the iterative inferences. This is because that HS-CAI($k=6$) can provide tight lower bounds to speed up the search so as to greatly reduce the constraint checks when solving the large scale problems under the limited memory budget.

Besides, we consider the DCOPs with the domain size of 3 and graph density of 0.6 as the dense configuration. The agent number varies from 14 to 24 and the average induced widths are 8 $\sim$ 18. 	
Fig. \ref{06agent} presents the performance comparison. 
It can be seen from Fig. \ref{06agent}(a) that DPOP and MB-DPOP cannot solve the problems with the agent number greater than 20 due to the large induced widths. Furthermore, since the inference-based complete algorithms have to perform inference on the entire solution space, these algorithms require much more NCLOs than the other competitors as Fig. \ref{06agent}(c) shows. 
Additionally, although they both perform the context-based inference, it can be seen from Fig. \ref{06agent}(b) and (c) that HS-CAI exhibits great superiority over MB-DPOP in terms of the network load and NCLOs. That is because HS-CAI only performs inference for the context patterns extracted by the context evaluation mechanism, while MB-DPOP needs to iteratively perform inference for all the contexts of cycle-cut variables. 
As for HS-CAI with different $k$, it can be seen from Fig. \ref{06agent}(a) and (c) that HS-CAI($k=10$) requires fewer messages but more NCLOs than HS-CAI($k=6$). That is because HS-CAI($k=10$) can produce tighter lower bounds but will incur more computation overheads than HS-CAI($k=6$).
\section{Conclusion}
By analyzing the feasibility of hybridizing search and inference, we propose a complete DCOP algorithm, named HS-CAI which combines search with context-based inference for the first time. Different from the existing hybrid complete algorithms, HS-CAI constructs tight lower bounds to speed up the search by executing context-based inference iteratively. Meanwhile, HS-CAI only needs to perform inference for a part of the contexts obtained from the search process by means of a context evaluation mechanism, which reduces the huge traffic overheads incurred by iterative context-based inferences.
We theoretically prove that the context-based inference can produce tighter lower bounds compared to the context-free approximated inference under the same memory budget. Moreover, the experimental results show that HS-CAI can find the optimal solution faster with less traffic overheads than the state-of-the-art.

In the future, we will devote to further accelerating the search process by arranging the search space with the inference results. In addition, we will also work for reducing the overheads caused by a context-based utility propagation. 
\section{Acknowledgments}
%We would like to thank the anonymous reviewers for their valuable comments and helpful suggestions. 
This work is funded by the Chongqing Research Program of Basic Research and Frontier Technology (No.:cstc2017jcyjAX0030), Fundamental Research Funds for the Central Universities (No.:2019CDXYJSJ0021) and Graduate Research and Innovation Foundation of Chongqing (No.:CYS17023). 
\bibliography{ref}  % put name of your .bib file here

\begin{thebibliography}{}

\bibitem[\protect\citeauthoryear{Atlas, Warner, and
  Decker}{2008}]{atlas2008memory}
Atlas, J.; Warner, M.; and Decker, K.
\newblock 2008.
\newblock A memory bounded hybrid approach to distributed constraint
  optimization.
\newblock In {\em Proceedings 10th International Workshop on DCR},  37--51.

\bibitem[\protect\citeauthoryear{Brito and Meseguer}{2010}]{Brito2010Improving}
Brito, I., and Meseguer, P.
\newblock 2010.
\newblock Improving {DPOP} with function filtering.
\newblock In {\em Proceedings of the 9th AAMAS},  141--148.

\bibitem[\protect\citeauthoryear{Chechetka and Sycara}{2006}]{chechetka2006no}
Chechetka, A., and Sycara, K.
\newblock 2006.
\newblock No-commitment branch and bound search for distributed constraint
  optimization.
\newblock In {\em Proceedings of the 5th AAMAS},  1427--1429.

\bibitem[\protect\citeauthoryear{Dechter, Cohen, and
  others}{2003}]{dechter2003constraint}
Dechter, R.; Cohen, D.; et~al.
\newblock 2003.
\newblock {\em Constraint processing}.
\newblock Morgan Kaufmann.

\bibitem[\protect\citeauthoryear{Deng \bgroup et al\mbox.\egroup
  }{2019}]{deng2019pt}
Deng, Y.; Chen, Z.; Chen, D.; Jiang, X.; and Li, Q.
\newblock 2019.
\newblock {PT-ISABB}: A hybrid tree-based complete algorithm to solve
  asymmetric distributed constraint optimization problems.
\newblock In {\em Proceedings of the 18th AAMAS},  1506--1514.

\bibitem[\protect\citeauthoryear{Farinelli \bgroup et al\mbox.\egroup
  }{2008}]{farinelli2008decentralised}
Farinelli, A.; Rogers, A.; Petcu, A.; and Jennings, N.~R.
\newblock 2008.
\newblock Decentralised coordination of low-power embedded devices using the
  max-sum algorithm.
\newblock In {\em Proceedings of the 7th AAMAS}, volume~2,  639--646.

\bibitem[\protect\citeauthoryear{Farinelli, Rogers, and
  Jennings}{2014}]{farinelli2014agent}
Farinelli, A.; Rogers, A.; and Jennings, N.~R.
\newblock 2014.
\newblock Agent-based decentralised coordination for sensor networks using the
  max-sum algorithm.
\newblock {\em Autonomous agents and multi-agent systems} 28(3):337--380.

\bibitem[\protect\citeauthoryear{Fioretto \bgroup et al\mbox.\egroup
  }{2017}]{fioretto2017distributed}
Fioretto, F.; Yeoh, W.; Pontelli, E.; Ma, Y.; and Ranade, S.~J.
\newblock 2017.
\newblock A distributed constraint optimization ({DCOP}) approach to the
  economic dispatch with demand response.
\newblock In {\em Proceedings of the 16th AAMAS},  999--1007.

\bibitem[\protect\citeauthoryear{Fioretto, Pontelli, and
  Yeoh}{2018}]{fioretto2018distributed}
Fioretto, F.; Pontelli, E.; and Yeoh, W.
\newblock 2018.
\newblock Distributed constraint optimization problems and applications: A
  survey.
\newblock {\em Journal of Artificial Intelligence Research} 61:623--698.

\bibitem[\protect\citeauthoryear{Fioretto, Yeoh, and
  Pontelli}{2017}]{fioretto2017multiagent}
Fioretto, F.; Yeoh, W.; and Pontelli, E.
\newblock 2017.
\newblock A multiagent system approach to scheduling devices in smart homes.
\newblock In {\em Proceedings of the 16th AAMAS},  981--989.

\bibitem[\protect\citeauthoryear{Freuder and Quinn}{1985}]{freuder1985taking}
Freuder, E.~C., and Quinn, M.~J.
\newblock 1985.
\newblock Taking advantage of stable sets of variables in constraint
  satisfaction problems.
\newblock In {\em Proceedings of the 9th IJCAI}, volume~85,  1076--1078.

\bibitem[\protect\citeauthoryear{Gershman, Meisels, and
  Zivan}{2009}]{gershman2009asynchronous}
Gershman, A.; Meisels, A.; and Zivan, R.
\newblock 2009.
\newblock Asynchronous forward bounding for distributed {COPs}.
\newblock {\em Journal of Artificial Intelligence Research} 34:61--88.

\bibitem[\protect\citeauthoryear{Gutierrez, Meseguer, and
  Yeoh}{2011}]{gutierrez2011generalizing}
Gutierrez, P.; Meseguer, P.; and Yeoh, W.
\newblock 2011.
\newblock Generalizing {ADOPT} and {BnB-ADOPT}.
\newblock In {\em Proceedings of the 22nd IJCAI},  554--559.

\bibitem[\protect\citeauthoryear{Hirayama and
  Yokoo}{1997}]{hirayama1997distributed}
Hirayama, K., and Yokoo, M.
\newblock 1997.
\newblock Distributed partial constraint satisfaction problem.
\newblock In {\em International Conference on Principles and Practice of
  Constraint Programming},  222--236.

\bibitem[\protect\citeauthoryear{Kim and Lesser}{2014}]{Kim2014DJAO}
Kim, Y., and Lesser, V.
\newblock 2014.
\newblock {DJAO}: a communication-constrained {DCOP} algorithm that combines
  features of {ADOPT} and {Action-GDL}.
\newblock In {\em Proceedings of the 28th AAAI},  2680--2687.

\bibitem[\protect\citeauthoryear{Litov and Meisels}{2017}]{litov2017forward}
Litov, O., and Meisels, A.
\newblock 2017.
\newblock Forward bounding on pseudo-trees for {DCOPs} and {ADCOPs}.
\newblock {\em Artificial Intelligence} 252:83--99.

\bibitem[\protect\citeauthoryear{Maheswaran \bgroup et al\mbox.\egroup
  }{2004}]{maheswaran2004taking}
Maheswaran, R.~T.; Tambe, M.; Bowring, E.; Pearce, J.~P.; and Varakantham, P.
\newblock 2004.
\newblock Taking {DCOP} to the real world: Efficient complete solutions for
  distributed multi-event scheduling.
\newblock In {\em Proceedings of the 3rd AAMAS}, volume~1,  310--317.

\bibitem[\protect\citeauthoryear{Maheswaran, Pearce, and
  Tambe}{2006}]{Maheswaran2006A}
Maheswaran, R.~T.; Pearce, J.~P.; and Tambe, M.
\newblock 2006.
\newblock A family of graphical-game-based algorithms for distributed
  constraint optimization problems.
\newblock In {\em Coordination of large-scale multiagent systems}. Springer.
\newblock  127--146.

\bibitem[\protect\citeauthoryear{Modi \bgroup et al\mbox.\egroup
  }{2005}]{modi2005adopt}
Modi, P.~J.; Shen, W.-M.; Tambe, M.; and Yokoo, M.
\newblock 2005.
\newblock {ADOPT}: Asynchronous distributed constraint optimization with
  quality guarantees.
\newblock {\em Artificial Intelligence} 161(1-2):149--180.

\bibitem[\protect\citeauthoryear{Netzer, Grubshtein, and
  Meisels}{2012}]{netzer2012concurrent}
Netzer, A.; Grubshtein, A.; and Meisels, A.
\newblock 2012.
\newblock Concurrent forward bounding for distributed constraint optimization
  problems.
\newblock {\em Artificial Intelligence} 193:186--216.

\bibitem[\protect\citeauthoryear{Ottens, Dimitrakakis, and
  Faltings}{2017}]{ottens2017duct}
Ottens, B.; Dimitrakakis, C.; and Faltings, B.
\newblock 2017.
\newblock {DUCT}: An upper confidence bound approach to distributed constraint
  optimization problems.
\newblock {\em ACM Transactions on Intelligent Systems and Technologyc}
  8(5):69.

\bibitem[\protect\citeauthoryear{Petcu and
  Faltings}{2005a}]{Petcu2005Approximations}
Petcu, A., and Faltings, B.
\newblock 2005a.
\newblock Approximations in distributed optimization.
\newblock In {\em International Conference on Principles and Practice of
  Constraint Programming},  802--806.

\bibitem[\protect\citeauthoryear{Petcu and Faltings}{2005b}]{petcu2005scalable}
Petcu, A., and Faltings, B.
\newblock 2005b.
\newblock A scalable method for multiagent constraint optimization.
\newblock In {\em Proceedings of the 19th IJCAI},  266--271.

\bibitem[\protect\citeauthoryear{Petcu and Faltings}{2006}]{petcu2006odpop}
Petcu, A., and Faltings, B.
\newblock 2006.
\newblock {ODPOP}: An algorithm for open/distributed constraint optimization.
\newblock In {\em Proceedings of the 21st AAAI},  703--708.

\bibitem[\protect\citeauthoryear{Petcu and Faltings}{2007}]{petcu2007mb}
Petcu, A., and Faltings, B.
\newblock 2007.
\newblock {MB-DPOP}: A new memory-bounded algorithm for distributed
  optimization.
\newblock In {\em Proceedings of the 20th IJCAI},  1452--1457.

\bibitem[\protect\citeauthoryear{Vinyals, Rodriguez-Aguilar, and
  Cerquides}{2009}]{vinyals2009generalizing}
Vinyals, M.; Rodriguez-Aguilar, J.~A.; and Cerquides, J.
\newblock 2009.
\newblock Generalizing {DPOP}: {DPOP}, a new complete algorithm for {DCOPs}.
\newblock In {\em Proceedings of the 8th AAMAS},  1239--1240.

\bibitem[\protect\citeauthoryear{Yeoh, Felner, and Koenig}{2010}]{yeoh2010bnb}
Yeoh, W.; Felner, A.; and Koenig, S.
\newblock 2010.
\newblock {BnB-ADOPT}: An asynchronous branch-and-bound {DCOP} algorithm.
\newblock {\em Journal of Artificial Intelligence Research} 38:85--133.

\bibitem[\protect\citeauthoryear{Zhang \bgroup et al\mbox.\egroup
  }{2005}]{zhang2005distributed}
Zhang, W.; Wang, G.; Xing, Z.; and Wittenburg, L.
\newblock 2005.
\newblock Distributed stochastic search and distributed breakout: properties,
  comparison and applications to constraint optimization problems in sensor
  networks.
\newblock {\em Artificial Intelligence} 161(1-2):55--87.

\end{thebibliography}
\bibliographystyle{aaai}
\end{document}